\documentclass[11pt]{article}
\usepackage[letterpaper, left=1in, right=1in, bottom=1in, top=1in]{geometry}

\usepackage[T1]{fontenc}
\usepackage{diagbox}
\usepackage{authblk}
\usepackage{amsthm}
\usepackage{amsmath,amssymb}
\usepackage{ascmac}
\usepackage{booktabs}
\usepackage{hyperref}
\usepackage{cleveref}
\usepackage{graphicx}
\usepackage{array}
\usepackage[algoruled,linesnumbered,algo2e,vlined]{algorithm2e}
\usepackage{url}
\usepackage{multirow}
\usepackage{multicol}
\usepackage{cases}
\usepackage[normalem]{ulem}
\usepackage{boites}

\newcommand{\rev}[1]{#1^R}
\newcommand{\Prefix}{\mathsf{Prefix}}
\newcommand{\Substr}{\mathsf{Substr}}
\newcommand{\Suffix}{\mathsf{Suffix}}
\newcommand{\occ}{\mathsf{occ}}
\newcommand{\MUS}{\mathsf{MUS}}
\newcommand{\MAW}{\mathsf{MAW}}
\newcommand{\MRW}{\mathsf{MRW}}
\newcommand{\EBF}{\mathsf{EBF}}
\newcommand{\MR}{\mathsf{MRep}}
\newcommand{\CDAWG}{\mathsf{CDAWG}}
\newcommand{\DAWG}{\mathsf{DAWG}}
\newcommand{\slink}{\mathsf{slink}}
\newcommand{\str}{\mathsf{str}}
\newcommand{\LPT}{\mathsf{LPT}}
\newcommand{\LPTplus}{\mathsf{LPT}^+}
\newcommand{\VCDAWG}{\mathsf{V}}
\newcommand{\ECDAWG}{\mathsf{E}}
\newcommand{\righte}{\mathsf{er}}
\newcommand{\lefte}{\mathsf{el}}
\newcommand{\emin}{\mathsf{e}_{\min}}

\newcommand{\children}{\mathsf{child}}

\newtheorem{theorem}{Theorem}
\newtheorem{corollary}{Corollary}
\newtheorem{lemma}{Lemma}

\theoremstyle{definition}

\begin{document}

\title{Computing Minimal Absent Words and Extended Bispecial Factors with CDAWG Space}

\author[1]{Shunsuke~Inenaga}
\author[2]{Takuya~Mieno}
\author[3]{Hiroki~Arimura}
\author[1]{Mitsuru~Funakoshi}
\author[1]{Yuta~Fujishige}

\affil[1]{Department of Informatics, Kyushu University, Japan}
\affil[2]{Department of Computer and Network Engineering, University of Electro-Communications, Japan}
\affil[3]{Graduate School of IST, Hokkaido University, Japan}

\date{}
\maketitle

\begin{abstract}
  A string $w$ is said to be a \emph{minimal absent word} (\emph{MAW})
  for a string $S$
  if $w$ does not occur in $S$ and any proper substring of $w$ occurs in $S$.
  We focus on non-trivial MAWs which are of length at least 2. 
  Finding such non-trivial MAWs for a given string is motivated by applications in
  bioinformatics and data compression.
  Fujishige et al. [TCS 2023] proposed
  a data structure of size $\Theta(n)$ that can output the set $\MAW(S)$ of all MAWs
  for a given string $S$ of length $n$
  in $O(n + |\MAW(S)|)$ time,
  based on the \emph{directed acyclic word graph} (\emph{DAWG}).
  In this paper, we present a more space efficient data structure
  based on the \emph{compact DAWG} (\emph{CDAWG}),
  which can output $\MAW(S)$ in $O(|\MAW(S)|)$ time with $O(\emin)$ space,
  where $\emin$ denotes the minimum of the sizes of the CDAWGs for
  $S$ and for its reversal $\rev{S}$.  
  For any strings of length $n$, it holds that $\emin < 2n$,
  and for highly repetitive strings $\emin$ can be sublinear (up to logarithmic) in $n$. 
  We also show that MAWs and their generalization \emph{minimal rare words} have close relationships with \emph{extended bispecial factors}, via the CDAWG.
\end{abstract}

\section{Introduction}

A string $w$ is said to be a \emph{minimal absent word} (\emph{MAW}) for a string $S$ if $w$ does not occur in $S$ and any proper substring of $w$ occurs in $S$.
Throughout this paper, we focus on MAWs which are of form $aub$,
where $a,b$ are characters and $u$ is a (possibly empty) string.
Finding such MAWs for a given string is motivated by applications including
bioinformatics~\cite{Almirantis2017MolecularBiology,Charalampopoulos18,pratas2020persistent,koulouras2021significant}
and data compression~\cite{Crochemore2000DCA,crochemore2002improved,AyadBFHP21}.

Fujishige et al.~\cite{FujishigeTIBT23} proposed
a data structure of size $\Theta(n)$ that can output the set $\MAW(S)$ of all MAWs for a given string $S$ of length $n$ in $O(n + |\MAW(S)|)$ time,
based on the \emph{directed acyclic word graph} (\emph{DAWG})~\cite{Blumer1985}.
The DAWG for string $S$, denoted $\DAWG(S)$, is the smallest DFA recognizing  all suffixes of $S$.

In this paper, we present a more space efficient data structure
based on the \emph{compact DAWG} (\emph{CDAWG})~\cite{Blumer1987},
which takes $O(\emin)$ space and can output $\MAW(S)$ in $O(|\MAW(S)|)$ time,
where $\emin = \min\{\lefte(S), \righte(S)\}$,
with $\lefte(S)$ and $\righte(S)$ being the numbers of edges of the CDAWGs for
$S$ and for its reversal $\rev{S}$, respectively.
For any string $S$ of length $n$ it holds that
the number of edges of the CDAWG for $S$ is less than $2n$
and it can be sublinear (up to logarithmic) in $n$
for highly repetitive strings~\cite{Rytter06,RadoszewskiR12,BelazzouguiC17}.

Our new data structure with the CDAWG is built on our deeper analysis on how
the additive $O(n)$ term is required in the output time of Fujishige et al.'s DAWG-based algorithm~\cite{FujishigeTIBT23}.
Let $x$ be a node of $\DAWG(S)$ and let $u$ be the longest string represented by $x$.
Suppose that there is a node $y$ of which the suffix link points to $x$,
and let $a$ be the character such that $au$ is the shortest string represented by $y$.
It follows that $aub$ is a MAW for $S$ if and only if
node $x$ has an out-edge labeled $b$ but node $y$ does not.
Fujishige et al.'s algorithm finds such characters $b$ by
comparing the out-edges of $y$ and $x$ in sorted order.
However, Fujishige et al.'s algorithm performs redundant comparisons
even for every pair of \emph{unary nodes} $x, y$ connected by a suffix link.
This observation has led us to the use of the CDAWG for $S$,
where unary nodes are compacted.
Still, we need to manage the case where $y$ is unary but $x$ is not unary.
We settle this issue by a non-trivial use of the \emph{extended longest path tree} of the CDAWG~\cite{Inenaga_LCDAWG_2024}.

In the previous work by Belazzougui and Cunial~\cite{BelazzouguiC17},
they proposed another CDAWG-based data structure
of size $O(\max\{\lefte(S), \righte(S)\})$ that can output all MAWs for $S$ in $O(\max\{\lefte(S), \righte(S)\}+|\MAW(S)|)$ time.
The $\lefte(S)$ term comes from the \emph{Weiner links} for all nodes of the CDAWG,
which can be as large as $\Omega(\righte(S)\sqrt{n})$ for some strings~\cite{Inenaga_LCDAWG_2024}.
Note that our data structures uses only $O(\emin) = O(\min\{\lefte(S), \righte(S)\})$ space and takes only $O(|\MAW(S)|)$ time to report the output.
We also present a simpler algorithm that restores
the string $u$ in linear time in its length using the CDAWG-grammar~\cite{BelazzouguiC17},
without using level-ancestor data structures~\cite{BerkmanV94,BenderF04},
which can be of independent interest.

We also show how our CDAWG-based data structure of $O(\emin)$ space can output
all \emph{extended bispecial factors} (\emph{EBFs}) \cite{AlmirantisCGIMP19} for the input string $S$
in output optimal time.
This proposed method can be seen as a more space-efficient variant of
Almirantis et al.'s method~\cite{AlmirantisCGIMP19}
that is built on suffix trees~\cite{Weiner1973}.
We also show that \emph{minimal rare words} (\emph{MRWs})~\cite{BelazzouguiC15},
which are generalizations of MAWs and \emph{minimal unique substrings} \emph{MUSs}~\cite{Ilie2011MUS},
have close relationships with EBFs through the CDAWG.

\section{Preliminaries}

\subsection{Basic Notations}
Let $\Sigma$ be an alphabet.
An element of $\Sigma$ is called a character.
An element of $\Sigma^\ast$ is called a string.
The length of a string $S$ is denoted by $|S|$.
The empty string $\varepsilon$ is the string of length 0.
Let $\Sigma^+ = \Sigma^* \setminus \{\varepsilon\}$.
If $S = xyz$, then $x$, $y$, and $z$ are called
a \emph{prefix}, \emph{substring} (or \emph{factor}), and \emph{suffix} of $S$, respectively.
Let $\Prefix(S)$, $\Substr(S)$, and $\Suffix(S)$ denote the sets of
the prefixes, substrings, and suffixes of $S$, respectively.
An element in $\Suffix(S) \setminus \{S\}$ is called a \emph{proper suffix} of $S$.
For any $1 \le i \le |S|$, $S[i]$ denotes
the $i$-th character of string $S$.
For any $1 \le i \le j \le |S|$, $S[i..j]$ denotes
the substring of $S$ that begins at position $i$ and ends at position $j$.
For convenience, let $T[i..j] = \varepsilon$ for any $i > j$.
We say that string $w$ \emph{occurs} in string $S$
iff $w$ is a substring of $S$.
Let $\occ_S(w) = |\{i \mid w = S[i..i+|w|-1]\}|$ denote
the number of occurrences of $w$ in $S$.
For convenience, let $\occ_S(\varepsilon) = |S|+1$.
For a string $S$, let $\rev{S} = S[|S|] \cdots S[1]$ denote the reversed string of $S$.

For a set $\mathsf{S}$ of strings, let $|\mathsf{S}|$ denote the number of strings in $\mathsf{S}$.

\subsection{MAW, MUS, MRW, and EBF}

A string $w$ is said to be an \emph{absent word} for
another string $S$ if $\occ_S(w) = 0$.
An absent word $w$ for $S$ 
is called a \emph{minimal absent word} (\emph{MAW}) for $S$ if $\occ_S(w[1..|w|-1]) \geq 1$ and $\occ_S(w[2..|w|]) \geq 1$.
By definition, any character $c$ not occurring in $S$ is a MAW for $S$.
In the rest of this paper, we exclude such trivial MAWs and 
we focus on non-trivial MAWs of form $aub$,
where $a,b \in \Sigma$ and $u \in \Sigma^*$.
We denote by $\MAW(S)$ the set of all non-trivial MAWs for $S$.

A string $w$ 
is called \emph{unique} in another string $S$ if $\occ_S(w) = 1$.
A unique substring $w$ for $S$
is called a \emph{minimal unique substring} (\emph{MUS}) for $S$
if $\occ_S(w[1..|w|-1]) \geq 2$ and $\occ_S(w[2..|w|]) \geq 2$.
Let $\MUS(S)$ denote the set of non-trivial MUSs of form $aub$ for string $S$.
Since a unique substring $w$ of $S$ has exactly one occurrence in $S$,
each $w$ element in $\MUS(S)$
can be identified with a unique interval $[i..j]$ such that
$1 \leq i \leq j \leq n$ and $w = S[i..j]$.

A string $w$ is called a \emph{minimal rare word} \emph{MRW} for
another string $S$ if
\begin{itemize}
  \item[(1)] $w = c \in \Sigma$, or
  \item[(2)] $w = aub$ with $a,b \in \Sigma$ and $u \in \Sigma^*$,
    such that $\occ_S(au) > \occ_S(aub)$ and $\occ_S(ub) > \occ_S(aub)$.
\end{itemize}
In particular, an MRW $w$ for $S$ is a MAW for $S$ when $\occ_S(w) = 0$.
Also, an MRW $w$ is a MUS for $S$ when $\occ_S(w) = 1$.
Let $\MRW(S)$ denote the set of non-trivial MRWs of type (2) for string $S$.
Then, $\MAW(S) \subseteq \MRW(S)$ and $\MUS(S) \subseteq \MRW(S)$ hold.

A string $w = aub \in \Substr(S)$ with $a,b \in \Sigma$ and $u \in \Sigma^*$
is called an \emph{extended bispecial factor} (\emph{EBF}) for string $S$
if $\occ_S(a'u) \geq 1$ and $\occ_S(ub') \geq 1$ for some characters $a' \neq a$ and $b' \neq b$.
Let $\EBF(S)$ denote the set of EBFs for string $S$.

\subsection{Maximal Substrings and Maximal Repeats}

A substring $u$ of a string $S$ is said to be \emph{left-maximal} in $S$ if
(1) there are distinct characters $a, b$ such that $au, bu \in \Substr(S)$, or
(2) $u \in \Prefix(S)$.
Symmetrically, a substring $u$ of $S$ is said to be \emph{right-maximal} in $S$ iff
(1) there are distinct characters $a, b$ such that $ua, ub \in \Substr(S)$, or
(2) $u \in \Suffix(S)$.
A substring $u$ of $S$ is said to be \emph{maximal} in $S$ iff
$u$ is both left-maximal and right-maximal in $S$.
A maximal substring $u$ is said to be a \emph{maximal repeat}
if $\occ_S(u) \geq 2$.
Let $\MR(S)$ denote the set of maximal repeats in $S$.
For a maximal repeat $u$ of a string $S$,
the characters $a \in \Sigma$ such that $\occ_S(au) \geq 1$
are called the \emph{left-extensions} of $u$,
and the characters $b \in \Sigma$ such that $\occ_S(ub) \geq 1$ are called the
\emph{right-extensions} of $u$.

The following relationship between MAWs and maximal repeats is known:
\begin{lemma}[\cite{BelazzouguiC15} and Theorem 1 of~\cite{PinhoFGR09}] \label{lem:maw_mr}
  For any MRW $aub$ for string $S$ with $a,b \in \Sigma$ and $u \in \Sigma^*$,
  $u$ is a maximal repeat in $S$.
\end{lemma}

\subsection{CDAWG}

\begin{figure}
  \centering
  \includegraphics[scale=0.5]{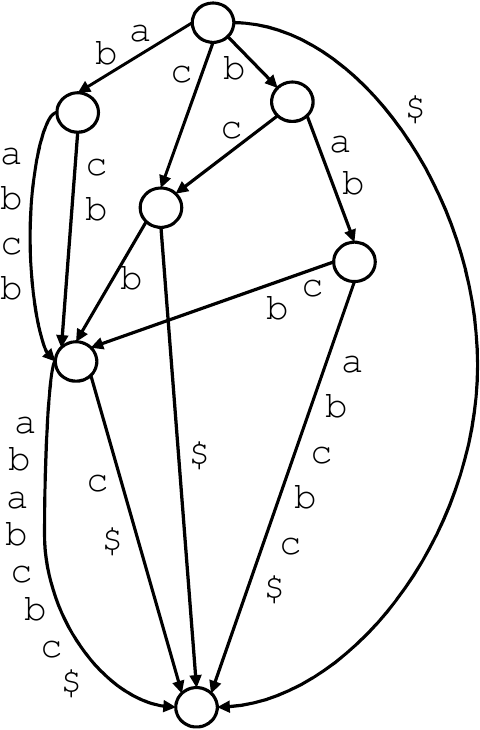}
  \caption{$\CDAWG(S)$ for string $S = \mathtt{ababcbababcbc\$}$.}
  \label{fig:CDAWG}
\end{figure}

The \emph{compact directed acyclic word graph} (\emph{CDAWG}) of string $S$,
that is denoted $\CDAWG(S) = (\VCDAWG, \ECDAWG)$,
is a path-compressed smallest partial DFA that represents $\Suffix(S)$,
such that there is a one-to-one correspondence between
the nodes of $\CDAWG(S)$ and the maximal substrings in $S$.
Intuitively, $\CDAWG(S)$ can be obtained by merging isomorphic subtrees of
the suffix tree~\cite{Weiner1973} for $S$.
Therefore, each string represented by a node $u$ of $\CDAWG(S)$ is a suffix of all the other longer strings represented by the same node $u$.
For convenience, for any CDAWG node $u \in \VCDAWG$,
let $\str(u)$ denote the \emph{longest} string represented by the node $u$.
The \emph{suffix link} of a CDAWG node $u$ points to another CDAWG node $v$
iff $\str(v)$ is the longest proper suffix of $\str(u)$ that is not represented by $u$.
We denote it by $\slink(u) = v$.
See Fig.~\ref{fig:CDAWG} for an example of $\CDAWG(S)$.

We remark that every internal node in $\CDAWG(S)$
corresponds to a maximal repeat in $S$,
that is, $\MR(S) = \{\str(v) \mid \mbox{$v$ is an internal node of $\CDAWG(S)$}\}$.
The first characters of the labels of the out-going edges of an internal node
of $\CDAWG(S)$ are the right-extensions of the corresponding maximal substring
for that node.
Hence the number of edges in $\CDAWG(S)$
is equal to the total number of right-extensions of maximal repeats in $\MR(S)$,
which is denoted by $\righte(S)$.
Similarly, we denote by $\lefte(S)$ the number of left-extensions of
the maximal repeats in $\MR(S)$.
Blumer et al.~\cite{Blumer1987} showed that $\min\{\righte(S),\lefte(S)\} < 2n$ for any string $S$ of length $n$.
On the other hand, it is known that there is an $\Omega(\sqrt{n})$ gap between $\righte(S)$ and $\lefte(S)$:

\begin{lemma}[Lemma 9 of~\cite{Inenaga_LCDAWG_2024}] 
  There exists a family of strings $S$ of length $n$
  such that $\righte(S) = \Theta(\sqrt{n})$ and $\lefte(S) = \Theta(n)$.
\end{lemma}

The \emph{length} of an edge $(u, y, v) \in \VCDAWG \times \Sigma^+ \times \VCDAWG$
in $\CDAWG(S)$ is its string label length $|y|$.
An edge $(u, y, v)$ is called a \emph{primary edge}
if $|\str(u)|+|y| = |\str(v)|$,
and $(u, y, v)$ is called a \emph{secondary edge}
otherwise ($|\str(u)|+|y| < |\str(v)|$).
Namely, the primary edges are the edges in the \emph{longest} paths
from the source to all the nodes of $\CDAWG(S)$.
Thus, there are exactly $|\VCDAWG|-1$ primary edges.
Let $\LPT(S)$ denote the spanning tree of $\CDAWG(S)$ that consists of the primary edges,
i.e., $\LPT(S)$ is the tree of the longest paths from the source to all nodes of $\CDAWG(S)$.

Since $\MR(\rev{S}) = \{\rev{w} \mid w \in \MR(S)\}$ holds for any string $S$,
the CDAWGs for $S$ and $\rev{S}$ can share the same nodes.
Let $\rev{\ECDAWG}$ denote the set of edges of $\CDAWG(\rev{S})$.
The primary (resp. secondary) edges in $\rev{\ECDAWG}$ are also called
\emph{hard} (resp. \emph{soft}) \emph{Weiner links} of $\CDAWG(S)$.
Each Weiner link is labeled by the first character of the label
of the corresponding edge in $\rev{\ECDAWG}$.
Note that $|\rev{\ECDAWG}| = \lefte(S)$.
\section{CDAWG Grammars}

A \emph{grammar-compression} is a method that represents an input string $S$
with a context-free grammar $\mathcal{G}$ that only generates $S$.
The size of $\mathcal{G}$ is evaluated by the total length of the right-hand sides of the productions,
which is equal to the size of the minimal DAG obtained by merging isomorphic subtrees of the derivation tree for $\mathcal{G}$.
Hence the more repeats $S$ has, the smaller grammar $S$ tends to have.

Assume that $S$ terminates with a unique character $\$$ not occurring elsewhere in $S$.
Belazzougui and Cunial~\cite{BelazzouguiC17} proposed a grammar-compression based on the CDAWGs.
We use a slightly modified version of their grammar:
The \emph{CDAWG-grammar} for string $S$ of length $n$, denoted $\mathcal{G}_\CDAWG$,
is such that: 
\begin{itemize}
  \item The non-terminals are the nodes in $\CDAWG(S)$;
  \item The root of the derivation tree for $\mathcal{G}_\CDAWG$ is the sink of $\CDAWG(S)$;
  \item The derivation tree path from its root to its leaf representing the $i$-th character $S[i]$ is equal to
    the reversed path from the sink to the source of $\CDAWG(S)$ that spells out the $i$-th (text positioned) suffix $S[i..n]$.
\end{itemize}
By definition, the DAG for the CDAWG-grammar $\mathcal{G}_\CDAWG$
is isomorphic to the reversed DAG for $\CDAWG(S)$, and hence,
the size for $\mathcal{G}_\CDAWG$ is $\righte(S)$.
See Fig.~\ref{fig:CDAWG_and_grammar} for an example.

Belazzougui and Cunial~\cite{BelazzouguiC17} showed how to extract the string label $x$ of a given CDAWG edge in $O(|x|)$ time
in a sequential manner (i.e. from $x[1]$ to $x[|x|]$),
by augmenting the CDAWG-grammar with a level-ancestor data structure~\cite{BerkmanV94,BenderF04} that works on the word RAM.
This allows one to represent $\CDAWG(S)$ in $O(\righte(S))$ space and perform pattern matching queries in optimal time,
without explicitly storing string $S$. They also used this technique to output a MAW $aub$ as a string: Since $u$ is a maximal repeat,
it suffices to decompress the path labels in $\LPT(S)$ from the source to the node $v$ that corresponds to $u$ (i.e. $\str(v) = u$).

We show that, when our goal is simply to obtain the longest string $\str(v)$ for a given CDAWG node $v$, then level-ancestor data structures are not needed.

\begin{figure}[t]
  \centering
  \includegraphics[scale=0.5]{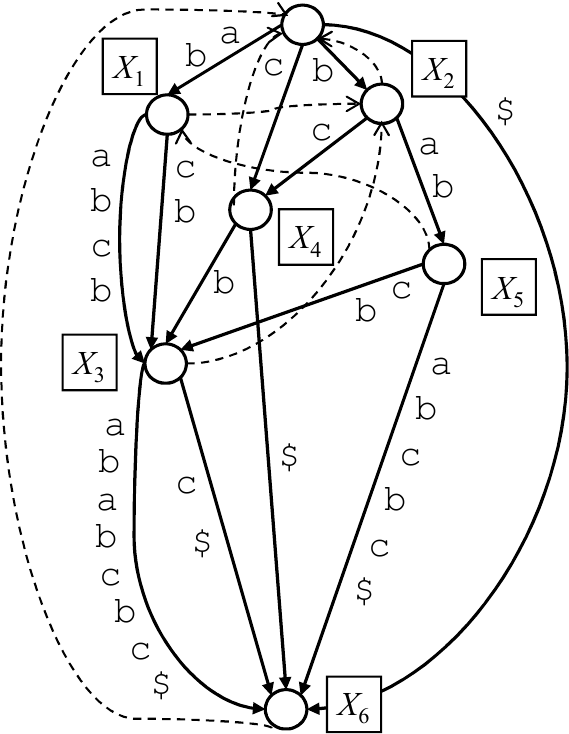}
  \hfill
  \includegraphics[scale=0.5]{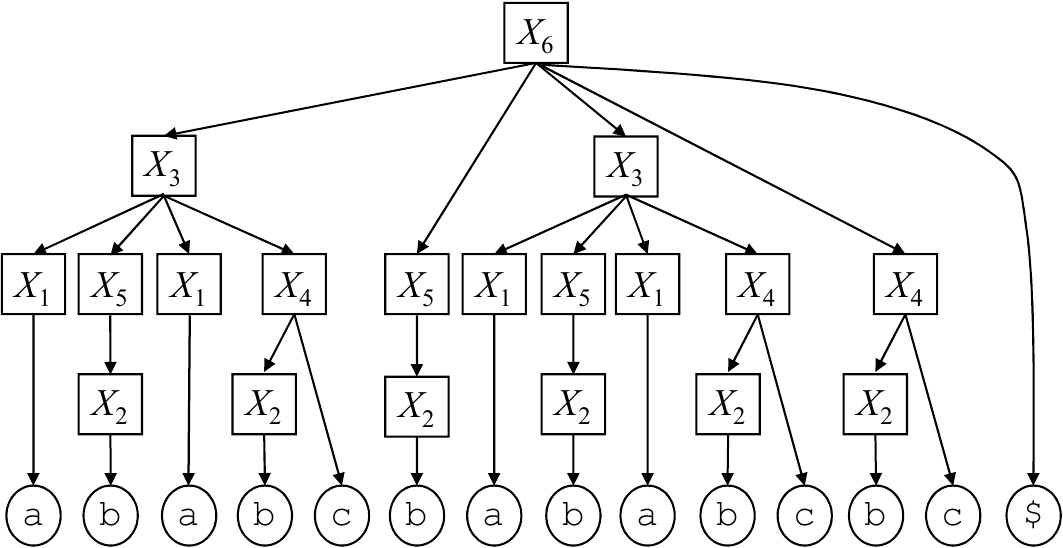}
  \caption{
    $\CDAWG(S)$ (left) and the CDAWG-grammar $\mathcal{G}_{\CDAWG}$ (right) for the running example from Fig.~\ref{fig:CDAWG}.
    The dashed arcs represent the suffix links of the nodes of $\CDAWG(S)$.
    To obtain $\str(X_3) = \mathtt{ababcb}$ for the CDAWG node $X_3$, we first decompress the non-terminal $X_3$ and obtain $\mathtt{ababc}$.
    We move to the node $X_2$ by following the suffix link of $X_3$. We then decompress the non-terminal $X_2$ and obtain the remaining $\mathtt{b}$.
  } \label{fig:CDAWG_and_grammar}
\end{figure}

\begin{lemma} \label{lem:decompression_suffix_link}
  Given $\CDAWG(S) = (\VCDAWG, \ECDAWG)$ with suffix links and grammar $\mathcal{G}_\CDAWG$ of total size $O(\righte(T))$,
  one can decompress the longest string $\str(v)$ represented by a given node $v$ in $O(|\str(v)|)$ time.
\end{lemma}

\begin{proof}
  For each CDAWG node $v \in \VCDAWG$, let $X(v)$ denote the non-terminal in the grammar $\mathcal{G}_\CDAWG$ that corresponds to $v$.
  We first decompress $X(v)$ and let $p_v$ be the decompressed string.
  Then, the length of $p_v$ is equal to the number of (reversed) paths between the source and $v$ in the CDAWG.
  Thus, if the suffix link $\slink(v)$ points to the source,
  then $p_v$ is the whole string $\str(v)$.
  Otherwise, then let $v_0 = v$.
  We take the suffix link $\slink(v_0) = v_1$,
  obtain string $p_{v_1}$ by decompressing $X(v_1)$.
  This gives us a longer prefix $p_{v_0} p_{v_1}$ of $\str(v)$,
  since in any interval $[i..i+|\str(v)|-1]$ of leaves of the derivation tree for $\mathcal{G}_\CDAWG$
  that represents $\str(v)$, $p_{v_1}$ immediately follows $p_{v_0}$.
  We continue this process until we encounter the last node $v_k$ in the suffix link chain from $v$, of which the suffix link points to the source.
  Then, the concatenated string $p_{v_0} p_{v_1}\cdots p_{v_k}$ is $\str(v)$.

  When there are no unary productions of form $X \rightarrow Y$,
  then the total number of nodes visited for obtaining $p_{v_0} p_{v_1}\cdots p_{v_k} = \str(v)$ is linear in its length.
  Otherwise, then for each maximal sequence of unary productions $X_1 \rightarrow X_2 \rightarrow \cdots \rightarrow X_\ell \rightarrow Y$,
  we create a pointer from $X_i$ to $Y$ for every $1 \leq i \leq \ell$ so that we can extract $\str(v)$
  in time linear in its length\footnote{A similar technique was used in the original CDAWG-grammar by Belazzougui and Cunial~\cite{BelazzouguiC17},
  where they explicitly removed such unary productions.}.
\end{proof}
\section{Computing MAWs and EBFs with CDAWG}

\subsection{Computing MAWs} \label{sec:MAW_CDAWG}

In this section, we first show the following:
\begin{theorem} \label{theo:MAW_CDAWG}
  There exists a data structure of $O(\emin)$ space
  that can output all MAWs for string $S$ in $O(|\MAW(S)|)$ time
  with $O(1)$ working space,
  where $\emin = \min\{\righte(S), \lefte(S)\}$.
\end{theorem}

For ease of description of our algorithm,
let us assume that
our input string $S$ begins and ends with unique terminal symbols
$\sharp$ and $\$$ not occurring inside $S$, respectively.
We also assume $\sharp, \$ \in \Sigma$ so that
they can be the first or the last characters for MAWs and EBFs
for $S$, respectively.
The case without $\sharp$ and $\$$ can be treated similarly.

We remark that MAWs have symmetric structures such that
$aub \in \MAW(S)$ iff $b\rev{u}a \in \MAW(\rev{S})$.
This implies that one can work with the smaller one of the two:
$\CDAWG(S)$ of $O(\righte(S))$ space
and $\CDAWG(\rev{S})$ of $O(\lefte(S))$ space.
In what follows, we assume w.l.o.g. that $\emin = \righte(S)$,
unless otherwise stated.

On top of the edge-sorted $\CDAWG(S)$ augmented with suffix links
and the CDAWG-grammar $\mathcal{G}_{\CDAWG}$ for $S$,
we use the extended version of $\LPT(S)$ proposed in~\cite{Inenaga_LCDAWG_2024},
that is denoted $\LPTplus(S)$.
Let $v$ be any node in the longest path tree $\LPT(S)$, and suppose that 
its corresponding node in $\CDAWG(S)$ has a secondary out-going edge $(v, u)$.
Then, we create a copy $\ell$ of $u$ as a leaf,
and add a new edge $(v, \ell)$ to the tree.
$\LPTplus(S)$ is the tree obtained by adding all secondary edges of $\CDAWG(S)$ to $\LPT(S)$ in this way, and thus, $\LPTplus(S)$ has exactly $\righte(S)$ edges.
For any node $v$ in $\LPTplus(S)$, let $\str(v)$ denote
the string label in the path from the root to $v$.

Let $u$ be any node in $\CDAWG(S)$
and suppose it has $k$ in-coming edges.
The above process of adding secondary edges to $\LPT(S)$
can be regarded as splitting $u$ into $k+1$ nodes $u', \ell_1, \ldots, \ell_k$,
such that $|\str(u')| > |\str(\ell_1)| > \cdots > |\str(\ell_k)|$.
We define the suffix links of these nodes such that
$\slink(u') = \ell_1$, $\slink(\ell_{i}) = \ell_{i+1}$ for $1 \leq i < k$,
and $\slink(\ell_k) = v'$ such that $v'$ is the node in $\LPTplus(S)$
that corresponds to the destination node $v = \slink(u)$ of the suffix link of $u$ in $\CDAWG(S)$.
We call these copied nodes $\ell_1, \ldots, \ell_k$ as \emph{gray nodes},
and the other node $u'$ as a \emph{white node}.
For convenience, we assume that the source is a white node.
See Fig.~\ref{fig:LPT-plus_MAW} for illustration.

\begin{figure}[t]
  \centering
  \includegraphics[scale=0.5]{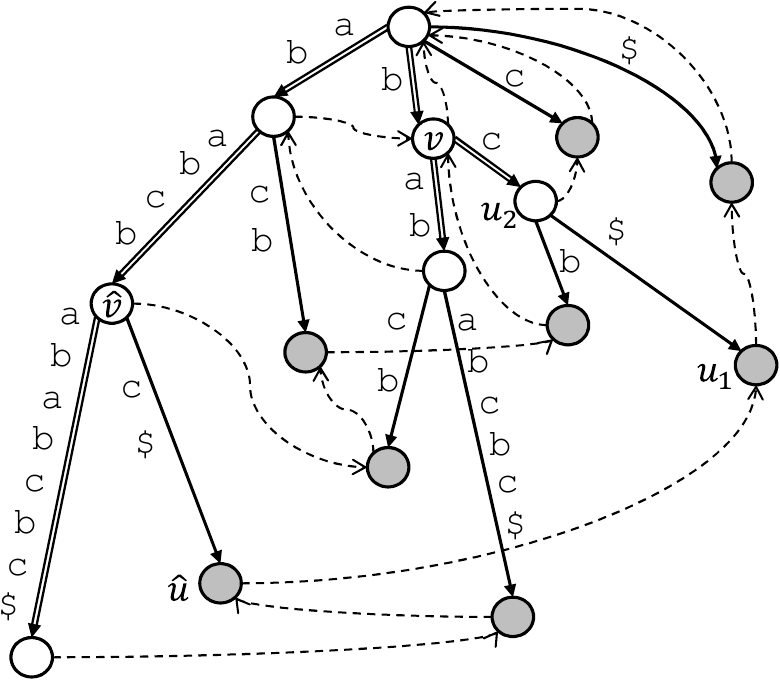}
  \caption{
    $\LPTplus(S)$ for the running example from Fig.~\ref{fig:CDAWG} and~\ref{fig:CDAWG_and_grammar}.
    The double-lined arcs represent primary edges, and the single-lined arcs represent secondary edges.
    For edge $(\hat{v}, \hat{u})$ with string label $\mathtt{c\$}$,
    consider the path $\langle v, u_1 \rangle$ that spells out $\mathtt{c\$}$ and is obtained by the fast link.
    The CDAWG node that corresponds to $u_2 = \mathtt{bc}$ has a virtual soft Weiner link with label $\mathtt{c}$
    pointing to the CDAWG node that corresponds to $\hat{u}$.
    Node $u_2$ has an out-edge with $\mathtt{b}$.
    Therefore, $\mathtt{c}u_2\mathtt{b} = \mathtt{cbcb}$ is a MAW for the string $S = \mathtt{ababcbababcbc\$}$.
  } \label{fig:LPT-plus_MAW}
\end{figure}

In what follows, we show how to compute MAWs for 
a given edge $(\hat{v}, \hat{u})$ in $\LPTplus(S)$.
Note that $\hat{v}$ is always a white node
since every gray node has no children.
Let $X$ be the string label of $(\hat{v}, \hat{u})$.
Let $v$ be the \emph{first} white node in the chain of suffix links from $\hat{v}$ and let $u$ be the descendant of $v$ in $\LPTplus(S)$
such that the path from $v$ to $u$ spells out $X$.
We maintain a \emph{fast link} from the edge $(\hat{v}, \hat{u})$
to the path $\langle v, u \rangle$ so we can access
from $(\hat{v}, \hat{u})$ to the terminal nodes $v$ and $u$
in the path in $O(1)$ time.
We have the four following cases (See also Fig.~\ref{fig:maw_cases}):
\begin{enumerate}
  \item[(A)] $\hat{u}$ and $u$ are both white nodes;
  \item[(B)] $\hat{u}$ is a white node and $u$ is a gray node;
  \item[(C)] $\hat{u}$ is a gray node and $u$ is a white node;
  \item[(D)] $\hat{u}$ and $u$ are both gray nodes.
\end{enumerate}

\begin{figure}[tb]
  \centering
  \begin{minipage}[t]{0.49\hsize}
    \rightline{
      \includegraphics[scale=0.4]{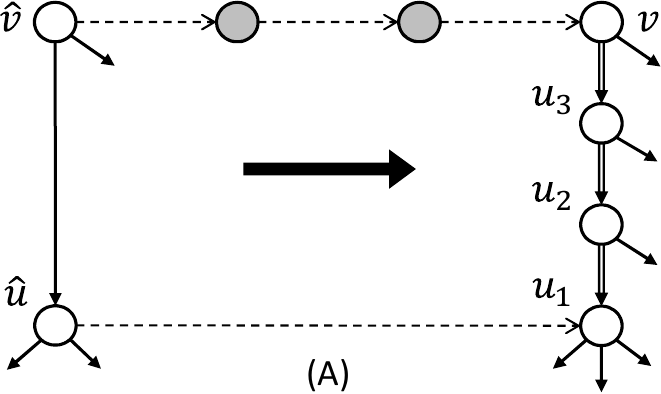}
    }
  \end{minipage}
  \hfill
  \begin{minipage}[t]{0.49\hsize}
    \rightline{
      \includegraphics[scale=0.4]{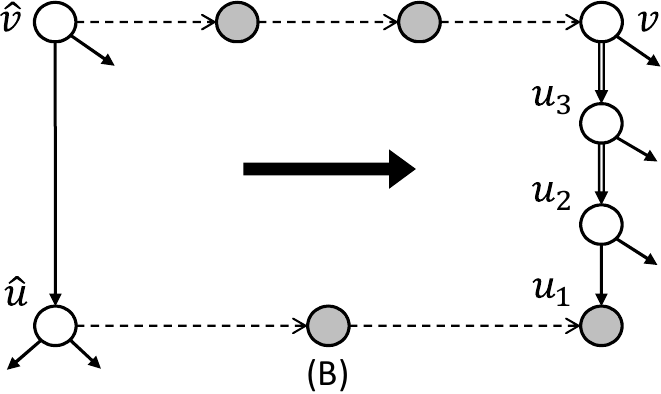}
    }
  \end{minipage} \vspace*{5mm} \\
  \begin{minipage}[t]{0.49\hsize}
    \rightline{
      \includegraphics[scale=0.4]{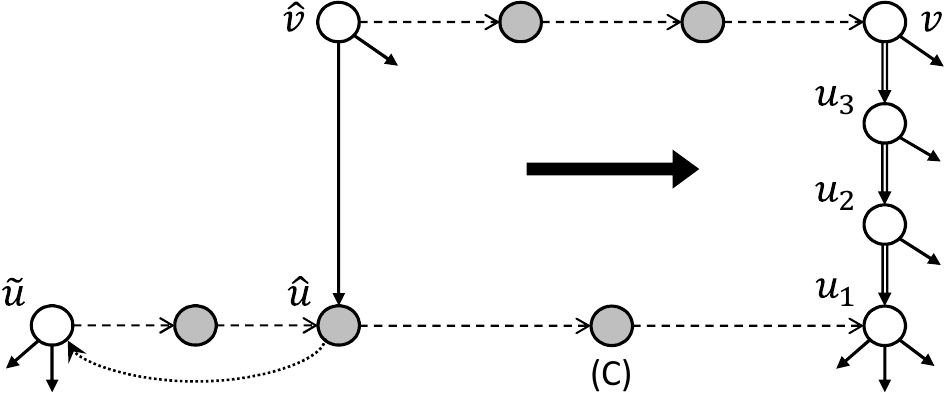}
    }
  \end{minipage}
  \hfill
  \begin{minipage}[t]{0.49\hsize}
    \rightline{
      \includegraphics[scale=0.4]{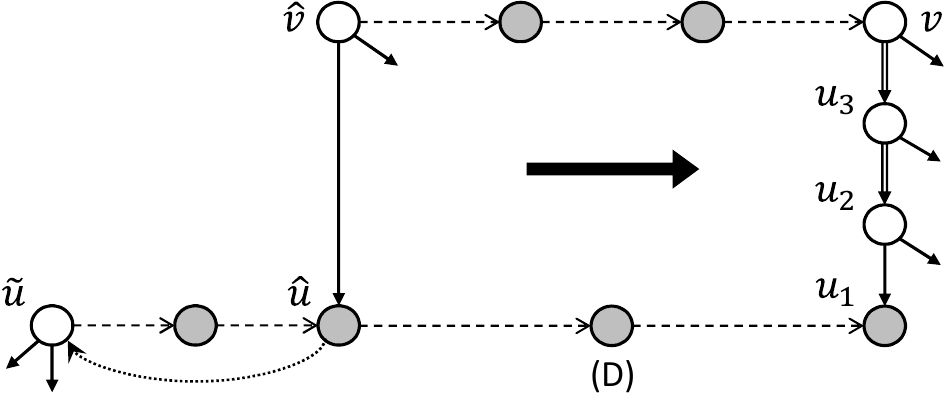}
    }
  \end{minipage}
  \caption{
    Cases (A), (B), (C), and (D) for computing MAWs from a given edge $(\hat{v}, \hat{u})$ on $\LPTplus(S)$.
    The bold arc represents the fast link from edge $(\hat{v}, \hat{u})$ to path $\langle v, u \rangle$, where $u = u_1$.
    The dashed arcs represent suffix links. The dotted arcs in Cases (C) and (D) are additional pointers for $O(1)$-time access from $\hat{u}$ to $\tilde{u}$.
  } \label{fig:maw_cases}
\end{figure}

\noindent{\textbf{How to deal with Case (A):}}
In Case (A), $u = \slink(\hat{u})$.

Let $m$ be the number of nodes in the path from
$v$ to $u$, with $v$ exclusive.
In the sequel, we consider the case where $m > 1$,
as the case with $m = 1$ is simpler to show.
Let $u_1~(= u), u_2, \ldots, u_m$ denote these $m$ nodes in the path
arranged in decreasing order of their depths.
We process each $u_i$ in increasing order of $i$.
For simplicity, we will identify each node $u_i$
as the string $\str(u_i)$ represented by $u_i$.

For $i = 1$,
node $u_1~(= u)$ has a hard Weiner link to $\hat{u}$.
Let $a$ be the character label of this hard Weiner link.
Let $\children(\hat{u})$ and $\children(\slink(\hat{u})) = \children(u_1)$ be the
sorted lists of the first characters of the out-edge labels from $\hat{u}$ and from $u_1$.
Note that $\children(\hat{u}) \subseteq \children(u_1)$.
Then, it follows from the definition of MAWs and Lemma~\ref{lem:maw_mr}
that $au_1b \in \MAW(S)$ iff $b \in \children(u_1) \setminus \children(\hat{u})$.

For each $i = 2, \ldots, m$,
observe that node $u_i$ has a (virtual) \emph{soft} Weiner link with the same
character label $a$, and leading to the same node $\hat{u}$.
We emphasize that this soft Weiner link is only virtual and is \emph{not} stored in our $O(\righte(S))$-space data structure.
We have the following lemma:
\begin{lemma} \label{lem:soft_Weiner_link_MAW}
  For each internal node $u_i$ with $2 \leq i \leq m$ in the path between $v$ and $u$,
  there exists at least one MAW $au_ib$ for $S$, where $a,b \in \Sigma$.
\end{lemma}
\begin{proof}
  Let $c$ be the first character of the label of the edge $(u_{i+1}, u_i)$.
  Let $\hat{u}_i$ denote the (virtual) implicit node
  on the edge $(\hat{v}, \hat{u})$, such that $\slink(\hat{u}_i) = u_i$.
  Note that $c$ is the unique character that immediately follows the locus of $\hat{u}_i$ in the string label of edge $(\hat{v}, \hat{u})$.
  Since $u_i$ is branching but $\hat{u}_i$ is non-branching,
  the set $\children(u_i) \setminus \children(\hat{u}_i) = \children(u_i) \setminus \{c\}$ is not empty.
  Thus, for any character $b \in \children(u_i) \setminus \{c\}$,
  we have $au_i b \in \MAW(S)$.
\end{proof}
Due to Lemma~\ref{lem:soft_Weiner_link_MAW}, we output MAWs
$au_ib$ for every character $b \in \children(u_i) \setminus \{c\}$.

\vspace*{0.5pc}
\noindent{\textbf{How to deal with Case (B):}}
In Case (B), $u$ is a gray node, and there might be other node(s) in the suffix link chain from $\hat{u}$ to $u$.
No MAW is reported for the gray node $u_1$~($= u$). The rest is the same as Case (A).

\vspace*{0.5pc}
\noindent{\textbf{How to deal with Cases (C) and (D):}}
In Cases (C) and (D) where node $\hat{u}$ is a gray node,
we use the white node $\tilde{u}$, which corresponds to the CDAWG node
from which $\hat{u}$ is copied.
We output MAW(s) by comparing out-edges of $\tilde{u}$ and $u_1$.
We store a pointer from $\hat{u}$ to $\tilde{u}$,
so that we can access $\tilde{u}$ from $\hat{u}$ in $O(1)$ time (see also Fig.~\ref{fig:maw_cases}).
The rest is the same as Cases (B) and (A).

\subsubsection{Putting all together.}
By performing the above procedure for all edges $(v, u)$ in $\LPTplus(S)$,
we can output all elements in $\MAW(S)$.

Let us analyze the time complexity and the data structure space usage.
In the case with $i \geq 2$,
by Lemma~\ref{lem:soft_Weiner_link_MAW},
accessing each node $u_i$ can be charged to any one of the MAWs 
reported with $u_i$.
The time required for $u_i$ is thus linear in the number of reported MAWs $au_ib$.
In the case with $i = 1$,
we compare the sorted lists $\children(\hat{u})$ and $\children(u_1)$.
Each successful character comparison with a character $b \in \children(u_1) \setminus \children(\hat{u})$ reports a MAW $au_1b$.
On the other hand, each unsuccessful character comparison
with a character $b \in \children(u_1) \cap \children(\hat{u})$
can be charged to the corresponding out-edge of $\hat{u}$
whose string label begins with $b$.
Thus, we can output all MAWs in $O(\righte(S)+|\MAW(S)|)$ time.
Further, when $|\MAW(S)| \leq \righte(S)$, then we can store
the set $\MAW(S)$ with $O(|\MAW(S)|) \subseteq O(\righte(S))$ space
and output each MAW in $O(1)$ time,
using Lemma~\ref{lem:decompression_suffix_link}.
This leads to our $O(|\MAW(S)|)$-time $O(\righte(S))$-space data structure.

To analyze the working memory usage,
let us consider the number of pointers.
For each given edge $(\hat{v},\hat{u})$,
we climb up the path from $\slink(\hat{u}) = u_1$ to $u_m$ one by one, until reaching $\slink(\hat{v}) = v$, which requires $O(1)$ pointers.
We also use $O(1)$ additional pointers for the character comparisons on the sorted lists of each pair of nodes.
Thus the total working space is $O(1)$.
We have proven Theorem~\ref{theo:MAW_CDAWG}.

\subsubsection{Combinatorics on MAWs with CDAWG.}
Theorem~\ref{theo:MAW_CDAWG} also leads to the following combinatorial property of MAWs:
\begin{theorem} \label{theo:num_MAWs}
  For any string $S$,
  $|\MAW(S)| = O(\sigma \emin)$.
  This bound is tight.
\end{theorem}

\begin{proof}
  For each node $u_i$ of $\LPTplus(S)$,
  the number of distinct characters $a$ such that $au_ib \in \MAW(S)$
  is at most the number of (hard or soft) Weiner links from $u_i$,
  which is bounded by $\sigma$.
  The number of distinct characters $b$
  such that $au_ib \in \MAW(S)$ does not exceed the number of out-edges from $u_i$.
  By summing them up for all nodes $u_i$ of $\LPTplus(S)$,
  we obtain $|\MAW(S)| = O(\sigma \cdot \righte(S))$.
  Since the same argument holds for the reversed string $\rev{S}$
  and $\LPTplus(\rev{S})$,
  we have $|\MAW(S)| = O(\sigma \cdot \min\{\righte(S), \lefte(S)\}) = O(\sigma \emin)$.

  It is known~\cite{MignosiRS02} that for the de Bruijn sequence
  $B_{k,\sigma}$ of order $k$ over an alphabet of size $\sigma$,
  $\MAW(B_{k,\sigma}) = \Omega(\sigma n)$, where $n = |B_{k,\sigma}|$.
  Since $\emin \leq \righte(S) < 2n$ for any string $S$ of length $n$,
  we have $\Omega(\sigma \emin)$ for $B_{k,\sigma}$.  
\end{proof}

\subsection{Computing EBFs}

We first prove the following combinatorial property for EBFs:
\begin{lemma} \label{lem:EBF_property}
  For any string $S$, $|\EBF(S)| \leq \righte(S)+\lefte(S)-|\VCDAWG|+1$.
\end{lemma}
\begin{proof}
  Recall our algorithm for computing $\MAW(S)$ in Section~\ref{sec:MAW_CDAWG}.
  Let $u_i$ be a node that has a (soft or hard) Weiner link with label $a \in \Sigma$ and has an out-edge starting with $b \in \Sigma$.
  Assuming $S[1] = \sharp$ and $S[|S|] = \$$,
  there are characters $a', b' \in \Sigma$
  such that $a' \neq a$, $b' \neq b$, and $a'u_i, u_ib' \in \Substr(S)$.
  Note also that $au_i, u_ib\in \Substr(S)$.
  Thus, if $au_ib \in \Substr(S)$, then $au_i b \in \EBF(S)$.
  If the locus for $au_i$ is the node $\hat{u}$,
  in which case the Weiner link from $u_i~(= u_1)$ to $\hat{u}$ is hard,
  then the EBF $au_ib$ is charged to the out-edge of $\hat{u}$ that starts with $b$.
  If the locus for $au_i$ is on the edge $(\hat{v},\hat{u})$,
  then the EBF $au_ib$ is charged to the soft-Weiner link from $u_i$ that labeled $a$.
  Since the number of hard Weiner links is $|V|-1$,
  the number of soft-Weiner links is $\lefte(S)-|V|+1$.
  Thus we have $|\EBF(S)| = \righte(S)+\lefte(S)-|V|+1$
  for any string $S$ with $S[1] = \sharp$ and $S[|S|] = \$$.

  If $S[1] \neq \sharp$ and/or $S[|S|] \neq \$$,
  then some node(s) $u_i$ in $\CDAWG(S)$ may only have
  unique characters $a,b$ such that $au_i, u_ib \in \Substr(S)$.
  Thus the bound is $|\EBF(S)| \leq \righte(S)+\lefte(S)-|\VCDAWG|+1$ for arbitrary strings $S$.
\end{proof}

By modifying our algorithm for computing MAWs for $S$
as is described in the proof for Lemma~\ref{lem:EBF_property},
we can obtain the following:

\begin{theorem} \label{theo:EBF_CDAWG}
  There exists a data structure of $O(\emin)$ space
  that can output all EBFs for string $S$ in $O(|\EBF(S)|)$ time
  with $O(1)$ working space.
\end{theorem}

\subsection{Computing length-bounded MAWs and EBFs}
\label{sec:length-bounded}

For any set $A \subseteq \Sigma^*$ of strings
and an integer $\ell > 1$,
let $A^{\geq \ell} = \{w \in A \mid |w| \geq \ell\}$
and $A^{\leq \ell} = \{w \in A \mid |w| \leq \ell\}$.

\begin{theorem} \label{theo:MAW_CDAWG_length}
  There exists a data structure of $O(\emin)$ space
  which, given a query length $\ell$, can output
  \begin{itemize}
    \item $\MAW(S)^{\geq \ell}$ in $O(|\MAW(S)^{\geq \ell}|+1)$ time,
    \item $\MAW(S)^{\leq \ell}$ in $O(|\MAW(S)^{\leq \ell}|+1)$ time,
    \item $\EBF(S)^{\geq \ell}$ in $O(|\EBF(S)^{\geq \ell}|+1)$ time,
    \item $\EBF(S)^{\leq \ell}$ in $O(|\EBF(S)^{\leq \ell}|+1)$ time,
  \end{itemize}
  with $O(1)$ working space.
\end{theorem}

\begin{proof}
  We only describe how to compute length-bounded MAWs.
  We can compute length-bounded EBFs similarly.

  In our data structure for reporting $\MAW(S)^{\leq \ell}$, we precompute
  the largest integer $m$ with
  $|\MAW(S)^{\leq m}| \leq \emin$,
  and store the elements in $\MAW(S)^{\leq m}$
  using Lemma~\ref{lem:decompression_suffix_link},
  sorted by their lengths.
  Given a query length $\ell$,
  if $\ell \leq m$, then we can output the elements of $\MAW(S)^{\leq \ell}$
  by reporting the stored MAWs in increasing order,
  in $O(|\MAW(S)^{\leq \ell}|+1)$ time.
  Otherwise (if $\ell > m$), then we process each edge $(\hat{v},\hat{u})$ in $\LPTplus(S)$ in increasing order of $|\str(\hat{u})|$.
  We need to access the internal nodes in the path from $\slink(\hat{v})$ to $\slink(\hat{u})$ topdown, in $O(1)$ time each.
  This can be done in $O(1)$ time per node by augmenting $\LPTplus(S)$ with the level-ancestor (LA) data structure~\cite{BerkmanV94,BenderF00}.
  We stop processing the edge $(\hat{v},\hat{u})$ as soon as we find the shallowest ancestor $u_i$ of $u = \slink(\hat{u})$ such that $|\str(u_i)| > \ell -2$.
  Since $|\MAW(S)^{\leq \ell}| > \emin$ for any $\ell < m$,
  this algorithm outputs $\MAW(S)^{\leq \ell}$ in $O(|\MAW(S)^{\leq \ell}|+1)$ time.

  Computing $\MAW(S)^{\geq \ell}$ is symmetric,
  except that the LA data structure is not needed:
  We can check the nodes between $\slink(\hat{u})$ and $\slink(\hat{v})$
  in decreasing order of their lengths, in $O(1)$ time each,
  by simply climbing up the path.
\end{proof}

Theorem~\ref{theo:MAW_CDAWG_length} is an improvement and a generalization over the previous work~\cite{Chairungsee2012PhylogenyByMAW}
that gave an $O(\ell n)$-space data structure that reports $\MAW(S)^{\leq \ell}$ in $O(\sigma n)$ time.
\section{Computing MRWs}
Let $\MRW_k(S) = \{w \in \MRW(S) \mid \occ_S(w) = k\}$.
Then $\MRW_0(S) = \MAW(S)$ and $\bigcup_{k \ge 0}\MRW_k(S) = \MRW(S)$ hold.
We show the following lemma.
Recall that any string in $\MAW(S)$ or $\EBF(S)$ is of length at least $2$.
\begin{lemma}\label{lem:mrw}
  For any string $S$ with $S[1] = \sharp$ and $S[|S|] = \$$,
  $\MRW(S)\setminus \EBF(S) = \MAW(S)$.
  In other words,
  $\MRW(S) \cap \EBF(S) = \bigcup_{k\ge1} \MRW_k(S)$.
\end{lemma}
\begin{proof}
  Clearly, $\MAW(S) \cap \EBF(S) = \emptyset$.
  It suffices to show that $\bigcup_{k\ge1} \MRW_k(S) \subseteq \EBF(S)$.
  Let $aub$ be an element in $\MRW_k(S)$ for some $k \ge 1$ where $a, b\in \Sigma$ and $u \in \Sigma^\star$.
  Since $aub$ is a minimal rare word occurring $k$ times in $S$,
  $\occ_S(au) \ge k+1$ and $\occ_S(ub) \ge k+1$.
  Also, since $k+1 \ge 2$, the last character of $au$ is not equal to $\$$, which is a unique character.
  Similarly, the first character of $ub$ is not equal to $\sharp$.
  Thus, any occurrence of $au$ (resp., $ub$) in $S$ is not a suffix (resp., prefix) of $S$.
  Hence, there must be characters $b' \ne b$ and $a' \ne a$
  such that $aub', a'ub \in \Substr(S)$.
  Therefore, $aub$ is an extended bispecial factor.
\end{proof}

Next, we estimate the number of minimal rare words occurring in $S$.
\begin{lemma}\label{lem:numofMRW}
  $|\bigcup_{k\ge1} \MRW_k(S)| \le \emin$.
\end{lemma}
\begin{proof}
  For any $w \in \bigcup_{k\ge1} \MRW_k(S)$,
  the locus of $w$ is the position on some edge $(v_1, v_2)$ 
  exactly one character down from $v_1$
  since $\occ_S(w[1..|w|-1]) > \occ_S(w)$.
  For the sake of contradiction, we assume that the locus of $w_1 = a_1u_1b \in \bigcup_{k\ge1} \MRW_k(S)$
  is the same as that of another $w_2 = a_2u_2b \in \bigcup_{k\ge1} \MRW_k(S)$
  where $a_1, a_2, b \in \Sigma$ and $u_1, u_2 \in \Sigma^\star$.
  We further assume $|w_1| \le |w_2|$ w.l.o.g.
  Since $a_1u_1$ and $a_2u_2$ reach the same CDAWG node, $a_1u_1$ is a suffix of $a_2u_2$.
  Since $w_1 \ne w_2$, $|w_1| < |w_2|$ holds, and thus, $w_1$ is a proper suffix of $w_2$.
  This contradicts that $\occ_S(w_2) = \occ_S(w_1)$ and $w_2$ is an MRW.
  Hence, the number of MRWs in $\bigcup_{k\ge1} \MRW_k(S)$ is never larger than $|\ECDAWG|$.
  The above argument holds even if the reversal of $S$ is considered.
  Therefore, the upper bound is $\min\{|\ECDAWG|, |\ECDAWG^R|\} = \emin$.
\end{proof}

The proof of Lemma~\ref{lem:numofMRW} implies that
if the locus of an MRW $aub$ is on an edge $(v_1, v_2)$, then $au$ is the \emph{shortest} string represented by the node $v_1$.
Namely, $\slink(v_1) = u$.
From this fact, we can represent each MRW $aub$ as its corresponding edge $(v_1, v_2)$
where
the character label of the hard Weiner link from $\slink(v_1)$ to $v_1$ is $a$,
the longest string represented by $\slink(v_1)$ is $u$, and
the first character of the edge label of $(v_1, v_2)$ is $b$.
Also, by Lemma~\ref{lem:decompression_suffix_link}, we can restore the MRW $aub$ from $(v_1, v_2)$
in time linear in its length.
Thus, the following lemma holds.
\begin{lemma}
  There exists a data structure of $O(\emin)$ space that can output all MRWs occurring in $S$
  in $O(|\bigcup_{k\ge 1}\MRW_k(S)|)$ time with $O(1)$ working space.
\end{lemma}
If $k$ is fixed in advance, we can store $\MRW_k(S)$ in $O(\emin)$ space and can output them in $O(|\MRW_k(S)|)$ time.
We obtain the next corollary by fixing $k=1$.
\begin{corollary}
  There exists a data structure of $O(\emin)$ space that can output all MUSs for string $S$
  in $O(|\MUS(S)|)$ time with $O(1)$ working space.
\end{corollary}
\section{Conclusions and Future Work}

In this paper, we proposed space-efficient $O(\emin)$-size data structures
that can output MAWs, EBFs, and MRWs for an input string $S$ of length $n$
in time linear in the output size,
where $\emin$ denotes the minimum of the left-extensions and the right-extensions
of the maximal repeats in $S$.
The key tools of our method are the $O(\emin)$-size grammar-representation of the CDAWG,
and the extended longest path tree of size $O(\emin)$.

Our future work includes the following:
Is it possible to extend our method for reporting length-bounded MAWs/EBFs of Section~\ref{sec:length-bounded}
to the problem of reporting MAWs/EBFs of length \emph{exactly} $\ell$ for a query length $\ell > 1$?
This would be a space-improvement of the previous work~\cite{charalampopoulos2018extended}.
Okabe et al.~\cite{OkabeMNIB23} proposed a (non-compacted) DAWG-based data structure
that can report the symmetric difference of the sets of MAWs for multiple strings, using $O(n)$ space and time linear in the output size,
where $n$ is the total length of the strings.
Can our CDAWG-based data structures be extended to this problem?

\section*{Acknowledgments}
This work was supported by JSPS KAKENHI Grant Numbers
JP20H05964, JP23K24808, JP23K18466~(SI),
JP23H04381~(TM),
JP20H00595, JP20H05963~(HA).

\end{document}